\documentclass[12pt]{article}
\usepackage{amsmath, amssymb, amsthm, amsfonts, mathtools, bm}
\usepackage{geometry}
\usepackage{physics}
\newtheorem{theorem}{Theorem}

\newtheorem{proposition}{Proposition}
\newtheorem{definition}{Definition}
\newtheorem{example}{Example}
\newtheorem{corollary}{Corollary}

\title{Quantum Simulation of Hyperbolic Equations and  the Nonexistence of a Dirac Path Measure}

\author{\bf Sumita Datta $^{1,2}$\\
$^1$ Alliance School of Applied Mathematics, Alliance University,\\ Bengaluru 562 106, India\\
$^2$ Department of Physics, University of Texas at Arlington,\\Texas 76019, USA\\}
\date{\today}

\begin{document}
\maketitle

\begin{abstract}
We revisit the longstanding issue of why no well-defined probability measure
exists corresponding to a classical (Kolmogorov) path integral representation
of the Dirac equation in Minkowski space. Two complementary perspectives are compared:
(i) Zastawniak's observation that the distributional character of the Dirac
propagator (presence of derivatives of the delta distribution) obstructs the
construction of a nonnegative transition kernel, and (ii) the indefinite
signature of the Minkowski metric which prevents positivity of the action
and yields oscillatory integrals. We show how these viewpoints can be unified
as different manifestations of a single mathematical obstruction from measure theoretical point of view, and we
discuss consequences for stochastic representations of relativistic first-order
equations.

\end{abstract}
\newpage
\section{Introduction}
In this work we emphasize on two complementary, measure-theoretic obstructions for constructing a path-space measure 
in the case of a Dirac-like hyperbolic equation:   
\begin{enumerate}
  \item A \emph{Minkowski-signature} obstruction: the indefinite Lorentzian
        metric leads to hyperbolic operators, oscillatory functionals $e^{\mathrm{i}S}$,
        and Euclidean actions whose exponentials are not positive, so no
        Feynman--Kac-type probabilistic representation exists.
  \item A \emph{Zastawniak-type} obstruction: the Dirac propagator is a
        distribution involving derivatives of the delta distribution, which
        cannot be realized as a nonnegative transition kernel on classical path
        space and is incompatible with Kolmogorov's extension theorem.
\end{enumerate}
Zastawniak originally formulated the problem in analytic terms: initial data
for the Dirac equation\cite{Dirac1928},\cite{ItzyksonZuber},\cite{PeskinSchroeder1995},\cite{Ryder1996} enter through derivative operators, and the associated
fundamental solutions are generalized functions rather than measures. Our aim
here is to recast and extend these results from a \emph{probabilistic} point
of view, making explicit the roles of Markov semigroups\cite{Simon2005},
\cite{GlimmJaffe1987}, Kolmogorov\cite{Kolmogorov1933} consistency,
Wiener geometry, and the \\ 
Bochner--Minlos\cite{Bochner1955,Minlos1963} framework.

We also discuss scalar (Klein--Gordon) fields\cite{BjorkenDrell}, which admit subordinated Brownian
representations, and hyperbolic telegrapher-type equations, which admit
velocity-jump process representations. These will later be used as benchmarks
for numerical simulations of hyperbolic equations. In contrast, fermionic
(Dirac) fields require Grassmann variables and Berezin integration rather than
classical probability measures.

In relativistic quantum theory the kinetic operator depends crucially on the
metric signature. The Minkowski\cite{Minkowski1910} signature $(+,-,-,-)$ produces hyperbolic
differential operators and oscillatory functional integrals of the form
$e^{\mathrm{i}S}$, whereas Wick rotation\cite{GlimmJaffe1987},\cite{StreaterWightman},\cite{Schwinger1951} $t\mapsto -\mathrm{i}\tau$ converts
the action to a Euclidean form $S_E$ and makes $e^{-S_E}$ suitable for
constructing probability measures in many scalar cases.

For bosonic, second-order parabolic operators one can often construct Gaussian
(Wiener) measures\cite{Wiener1923},\cite{Doob1953} and Feynman--Kac-type representations\cite{Feynman1948},\cite{Feynman1965},\cite{Donsker1950}, \cite{Kac1959}. For first-order
relativistic Dirac-type equations, several obstructions appear:
\begin{itemize}
  \item the propagator is a distribution involving derivatives of $\delta$;
  \item the Euclidean action does not give a positive weight $e^{-S_E}$;
  \item spinor fields are fermionic and require Grassmann variables;
  \item the hyperbolic character is incompatible with Brownian geometry.
\end{itemize}

Zastawniak's work \cite{Zastawniak1988,Zastawniak1989,Zastawniak1990,Zastawniak1994},
and other work along this line\cite{DeWittMoretteFoong1989},\cite{Brzeźniak1989}
emphasized analytic aspects: fundamental solutions for the Dirac equation are
generalized functions acting on initial data through derivatives, and no
path-space measure exists whose finite-dimensional marginal densities generate
the Dirac propagator. In this paper we give a probabilistically oriented
presentation of this nonexistence result, organized around:
\begin{enumerate}
	\item a \emph{probability-theoretic framework}: Markov semigroups, Feller\cite{FellerV1}
        processes, and Kolmogorov's extension theorem;
  \item the \emph{Minkowski and Euclidean} structures;
  \item the \emph{distributional} structure of Dirac propagators and its
        conflict with classical transition kernels.
\end{enumerate}

In later work, telegrapher equations\cite{Kac1974}, \cite{GriegoHersh1971}, 
\cite{GriegoHersh1969},\cite{Hersh1974},\cite{Hersh1975},\cite{Hersh2006},\cite{Gaveau1984} and the Kac velocity-jump process will
serve as a positive benchmark: they are hyperbolic, admit stochastic
representations, and can be simulated numerically. The contrast to the Dirac
case highlights the genuinely fermionic and non-probabilistic nature of Dirac
path integrals.\\

The organization of the paper is as follows.
Section~1 introduces the problem concerning the nonexistence of a Dirac path measure.
Section~2 presents the probabilistic framework employed to analyze this nonexistence problem.
Sections~3 and~4 discuss, respectively, the Minkowski-space obstruction and the Zastawniak obstruction
to the existence of a Dirac path measure.
Section~5 compares the measure-theoretic issues arising for scalar fields with those for fermionic fields.
Section~6 formulates a unified no-go theorem addressing the nonexistence of such measures.
Section~7 relates the present results to earlier work in the literature.
The paper concludes in Section~8.

\begin{table}[h!]
\begin{center}
\caption{\bf Notation Table}
\begin{tabular}{ccc}
\hline
Notation/Phrase & Meaning \\
$L$   & Banach space   \\
	$\tilde{T}(t)$     & A strongly continuous semigroup    \\
	$ A $  & Infinitesimal generator of $\tilde{T}(t)$ \\           
	$\mathsf{D}(A)$     & Domain of A   \\
$E $  & State space   \\
$S$ & Action  \\
$ T$  & index set for time t  \\
$X(t) $   & Stochastic Process  \\
$\mathbb{E}_x $  & Expectation value \\
$\mathbb{P}$ & Probability measure \\
	$\mathcal{C}_{(0,T)}$ & the classes of cylinder subset  \\
	${\mu}_{w}^{x_0} $ & Wiener measure on a Brownian path staring at $x_0$ \\
$ B_t$ & Brownian motion  \\
$ \mathbb{D}_E$ & Dirac operator in Euclidean space \\
$ \eta_{\mu\nu}$ & Minkowski metric \\
$ {\sigma}^j$ & Pauli spin matrices \\
$ {\gamma}^j $ & gamma matrices \\
$(\Omega,\mathcal{F},\mathbb{P})$ & Probability space \\
$ \Omega $ & Sample space \\
$ \mathcal{F} $& sigma algebra on $ \Omega $ \\
	$  (E, \mathcal{B})$ & complete separable metric spcae \\
$ E $ & Complete sample space \\
$ \mathcal{B} $ & sigma algebra on $E$ \\
$D_k$ & subset of  $D\subset\mathbb{R}^d$.\\
	$\mathcal{D}$ & Berezin integral-measure\\

\hline
\end{tabular}
\end{center}
\end{table}

\newpage
\section{Probabilistic Framework: Kolmogorov, Markov Semigroups and Feynman--Kac}
\label{sec:prob-framework}
We first recall the basic probabilistic structures that underlie stochastic
representations of PDEs, in order to formulate precisely what it would mean
for the Dirac equation to admit a \emph{probability measure} on path space.
\subsection{The connection between Semigroup Theory and Probability Theory}
\begin{definition}
	Let $ L $ be a Banach space. For each $ t \ge 0 $ let $ \tilde{T}(t): L \rightarrow L $ be a bounded linear operator
	(i.e.$  \parallel \tilde{T}(t)f \parallel \le M(t)\parallel f \parallel $ for each $ f $ in $ L $, and some finite constant $ M(t) ) $  then, \\
	the family $ \{\tilde{T}(t), t \ge 0 \} $ is a strongly continuous semigroup on $ L $ if the following conditions hold. \\
	(i) $ \tilde{T}(t+s)=\tilde{T}(t)\tilde{T}(s) $ for every $ s,t \ge 0 $ \\
	(ii) $\tilde{T}(0)=I $(identity operator) \\
	(iii) the mapping $ t \rightarrow \tilde{T}(t)f $ is continuous on $ [0, \infty ) $ for each  $ f $ in $ L $. \\
If in addition one has:\\
	$  \parallel \tilde{T}(t)f \parallel \le \parallel f \parallel $ for each $t \ge 0 $ and $ f $ in $ L $, then the semigroup $ \{\tilde{T}(t), t \ge 0 \} $ is called  
a strongly continuous contraction semigroup. 
\end{definition}
\begin{definition}
	Let $ \{\tilde{T}(t), t \ge 0 \} $ be a strongly continuous semigroup on $ L $. The infinitesimal generator $ A $ of $ \{\tilde{T}(t), t \ge 0 \} $ defined by \\
\[
	Af = \lim_{t\downarrow0}\frac{\tilde{T}(t)f-f}{t}.
\]
\end{definition}
The set of $f$ in $ L $ for which $ Af $ is defined is denoted by $\mathsf{D}(A)$ and called the domain of $ A $.\\
{\bf Fact 1} Let $A$ be the inifinitesimal generator of a strongly continuous contraction semigroup $ \{\tilde{T}(t), t \ge 0 \} $ on the Banach space $ L $ . Then 
\[
	\frac{d\tilde{T}(t)}{dt}=A \tilde{T}(t)f=\tilde{T}(t)Af
\]
Then $\mathsf{D}(A)$ is invariant with respect to each $ \tilde{T}(t) $. \\
{\bf Fact 2 ( Uniqueness Theorem)}
Let $ A $ be the the inifinitesimal generator of a strongly continuous contraction semigroup $ \{\tilde{T}(t), t \ge 0 \} $ on the Banach space $ L $. \\
If $ f \in \mathsf{D}(A) $, then $ u(t)=\tilde{T}(t)f $ is the unique solution of \\
\[
	\frac{du(x,t)}{dt}=Au(x,t)
	u(0)=u(0+)=f
\]
\begin{definition}
	Let $ (E,\mathcal{B}) $ be a complete separable metric space with $ \sigma $- algebra of Borel subsets of $ \mathcal{B} $ of $E$  and $ (\Omega, \mathcal{F}) $
	be a  measurable space. Let $ X(t, .): \Omega \rightarrow E $ be measurable with respect $ (E,\mathcal{B}) $  and $ (\Omega, \mathcal{F}) $. Then the collection 
	$ X(t) $ is called a stochastic process. \\
	{\bf Markov property:} For each $ s,t \ge 0 $, $x \in E $  and $ B \in \mathcal{B} $,
	\[
		P_{x}(X(t+s) \in B | \mathcal{F}_{s})= P_{X(s)}(X(t) \in B) 
\]
\end{definition}
\subsubsection{Markov semigroups and generators}
Let ${X(t)}_{t\ge0}$ be a Markov process on a state space $E$ with transition
probabilities
\[
	P_t(x,A) = \mathbb{P}(X(t)\in A\mid X(0)=x).
\]
For bounded measurable $f$, the Markov semigroup acts as
\[
	(P_t f)(x) = \mathbb{E}_x[f(X(t))] = \int_E f(y)P_t(x,dy),
\]
and satisfies
\[
P_0 = \mathrm{Id},\qquad P_{t+s}=P_tP_s,\qquad P_t 1 = 1,\qquad P_t f\ge0\ \text{if}\ f\ge0.
\]
The (infinitesimal) generator is defined on a suitable core $\mathsf{D}(A)$ by
\[
Af = \lim_{t\downarrow0}\frac{P_tf-f}{t}.
\]
In many classical examples (diffusions, jump processes) $A$ is a second-order
or integro-differential operator.

For PDEs of the form
\[
\partial_t u = Au,
\]
a Markov process with generator $A$ yields a probabilistic representation
\[
u(t,x) = (P_t f)(x) = \mathbb{E}_x[f(X(t))]
\]
for $u(0,\cdot)=f$. The key properties are: positivity, contraction
(e.g.\ on $L^\infty$ or $C_b$), and the semigroup law.

\subsection{Kolmogorov extension theorem}
To understand the Feynman-Kac formalism, a few words about the Wiener measure are in order. Since the existence Wiener measure is the cornerstone of the mathematically rigorous approach toward  path integrals, some highlights outlining its construction  are given below.

By the result due to Kolmogorov, any system of finite dimensional distributions $ {\mu}_{{t_1},........,{t_n}} $ , i.e., the family of probability measures, on $ R^n $ indexed by
$0\le{t_1<......t_n<t}$ and satisfying the following consistency conditions: \\

$K_1. \mu_{t_{\sigma(1)}........,t_{\sigma(n)}}(I_1 \times ........\times{I_n})={\mu}_{t_1,........,t_n}(I_{{\sigma}^{-1}(1)} \times ......\times I_{{\sigma}^{-1}(n)})$ for all permutations $\sigma$ on $\{1,2,......,n\}$.\\

$K_2. {\mu}_{{t_1},........,{t_n}}(I_1 \times ........\times{I_n})={\mu}_{{t_1},........,{t_k},{t_{k+1}}.....{t_{n+m}}}( I_1 \times ........\times{I_n}, \times R \times R.....\times R) $ for all $m \in \mathcal{N} $ where $I_1={[a_i, b_i]}, i=1,2,.....,n$ there exists a unique probability measure ${\mu}^{x_0}$ on $ R^{[0,T]}$ that has $\{{\mu}_{{t_1},........,{t_n}}\}$ as their finite dimensional distributions.\\

 More precisely, let $\mathcal{C}_{(0,T)}$ be the classes of cylinder subsets in $R^{[0,T]}$ whose path elements have the form $C_{I_1  ........ I_n}^{{t_1},........,{t_n}}=\{\Phi \in R^{[0,T]}|\Phi(t_1) \in I_1........\Phi(t_n) \in I_n ; 0 \le t_1<t_2........t_n \le T \}$(Figure 1).
\begin{figure}[h!]
\includegraphics[width=6in,angle=-0]{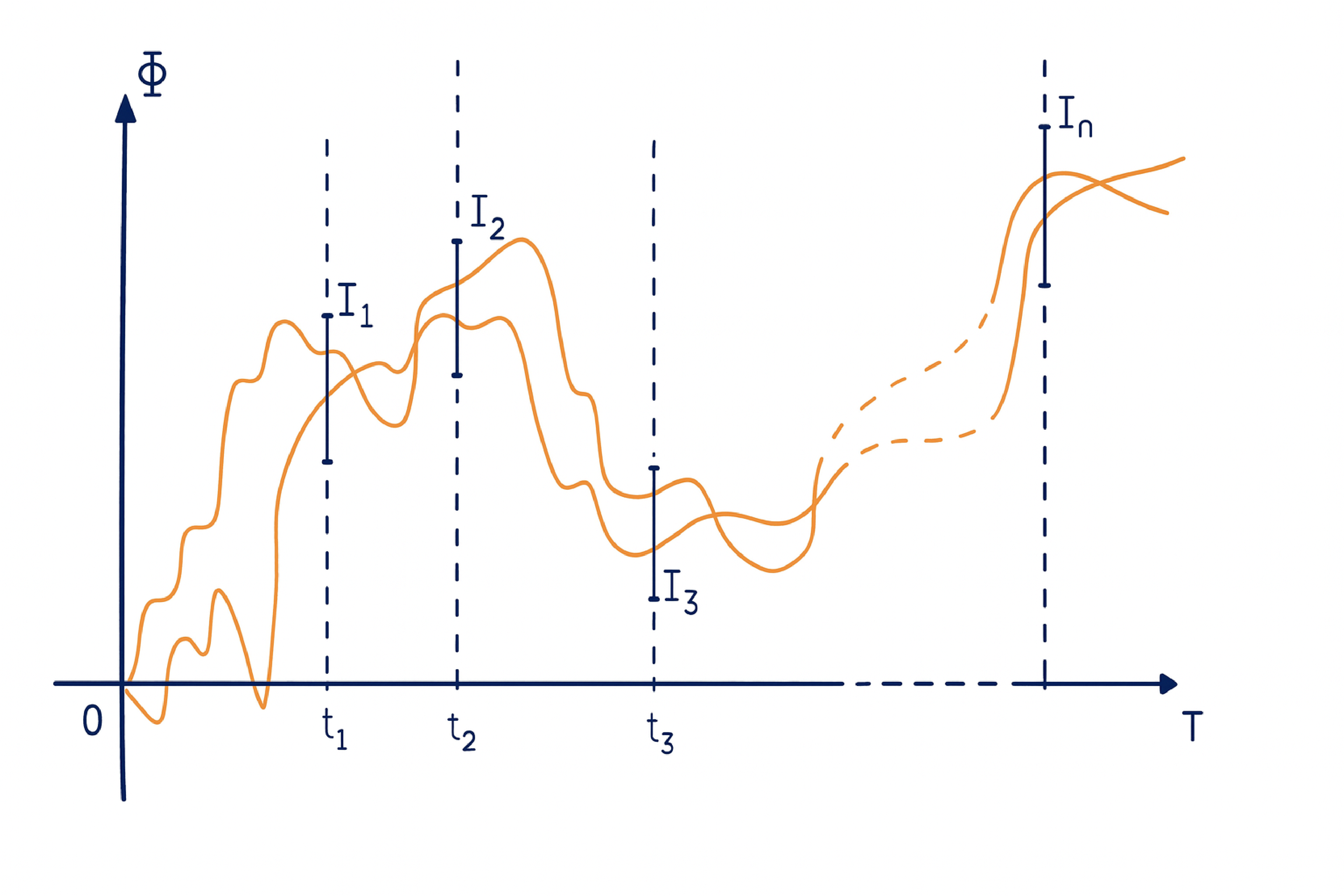}
\caption{The cylindrical subsets of $\Phi(t)$}
\end{figure}

 Then ${\mu}^{x_0}$ defined below for every $0 \le t_1<t_2........t_n \le T$, $I_1  ........ I_n \subset R $, ${\mu}^{x_0}(C_{I_1  ........ I_n}^{{t_1},........,{t_n}})=\int_{I_1} dx_1.............\int_{I_n} dx_n k(t/n;x_0,x_1)........k(t/n;x_{n-1},x_n) $ satisfies Kolmogorov's consistency conditions $(K_1),(K_2)$ and it can be extended from $\mathcal{C}_{(0,T)}$ to $\mathcal{B}_{(0,T)}$=$\sigma$-algebra of Borel subset in $ R^{[0,T]} $ by Kolmogorov's extension theorem.\\
 This extension of ${\mu}^{x_0}$ to $R^{[0,T]}$ is proven to be entirely supported by the space of continuous functions $\mathcal{C}_{(0,T)}$ and is called a Wiener measure ${\mu}_{w}^{x_0}$. As a result one dimensional Brownian motion started at $x_0$, $ \{ X(t), X_0=x_0,0 \le T$ can be identified with the following probability space
$({C}_{(0,T)}^{x_0},\mathcal{B}_{(0,T)}^{x_0},{\mu}_{w}^{x_0})$ where $C_{(0,T)}^{x_0}= \{ \Phi(t) \in C_{(0,T)}|\Phi(0)=x_0,0 \le t \le T \}$ and $P(X(t) \in I_1, X(t_1) \in I_2......X(t_n) \in I_n|X(0)=x_0)={\mu}_{w}^{x_0}C_{I_1  ........ I_n}^{{t_1},........,{t_n}}$.
Here $C_{(0,T)}^{x_0}$ is the space of all possible trajectories of Brownian motion originating at $x_0$, and Wiener measure ${\mu}_{w}^{x_0}$ prescribes probabilities to various set of trajectories.\\
Typical sets include $\{ \Phi \in C_{(0,T)}|f(t)<\Phi(t)<g(t);t \in [0,T]\}$ for any given functions $f(t), g(t) \in C_{(0,T)}(Figure 2)$.
\begin{figure}[h!]
\includegraphics[width=6in,angle=-0]{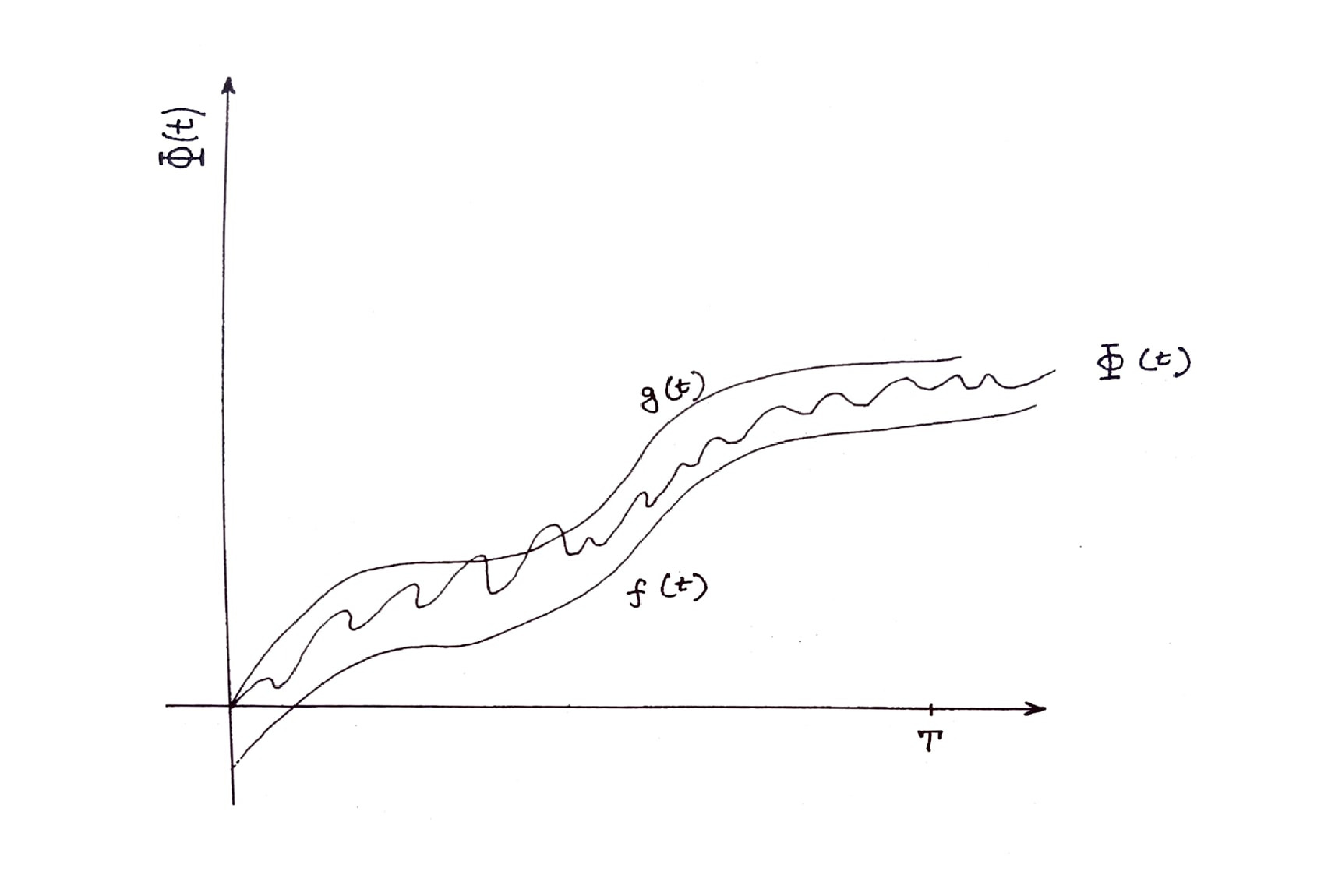}
\caption{A plot for the Brownian trajectories}
\end{figure}
Going back to the real valued Trotter product formula,\\
 $\tilde{\mu}_n^{x_0}=Proj_{R^{t/n} \times R^{{2t}/n} \times .......R^{{n-1}/n} \times R^t} {\mu}_{w}^{x_0}$
converges to $ Proj_{R^{0,T}} {\mu}_{w}^{x_0}={\mu}_{w}^{x_0}$ as $n \rightarrow \infty $ and thus one captures the Wiener measure in the limit. The construction of Brownian motion and Wiener measure extend readily to $R^n$.\\
To construct a probability measure on path space, one typically specifies
finite-dimensional distributions
\[
\mathbb{P}(X(t_1)\in I_1,\dots,X(t_n)\in I_n)
\]
or, in density form,
\[
p_{t_1,\dots,t_n}(x_1,\dots,x_n),
\]
and requires:
\begin{itemize}
  \item nonnegativity and normalization,
  \item consistency under marginalization in time,
  \item appropriate measurability and regularity.
\end{itemize}
Kolmogorov's extension theorem \cite{Kolmogorov1933,BogachevMeasure,RudinFA,YosidaFA} then
guarantees the existence of a probability measure on path space whose
finite-dimensional marginals are the given $p_{t_1,\dots,t_n}$.

From this perspective, asking for a Dirac ``path measure'' means asking for
a family of nonnegative densities $p_{t_1,\dots,t_n}$ whose marginals produce
the Dirac propagator and satisfy the Kolmogorov consistency conditions.

\subsection{Feynman--Kac for parabolic equations}
{ Extension of the fractional heat equation in the interacting system by the Feynman-Kac formula
\cite{Datta1996,Datta2000,Datta2023,Datta2024,Datta2025,Datta2022},\cite{Korzeniowski1992}:}
\begin{proposition}
	Let the $Schr\ddot{o}dinger$ semigroup[50] $\tilde{T}(t)=e^{tH_0}$ be defined by $\tilde{T}(t)f(x)=E_xf(X(t))$ for f in $L$ where $L$ is the Banach space of bounded real valued functions with generator of the semigroup $ A=H_0=C_{2}{\Delta}$. Then
	$\tilde{\tilde{T}}(t)f(x)=E_x {e^{-\int_{0}^{t}V(X(s))ds}f(X(t)} $ defines a semigroup with infinitesimal generator $\tilde A$ such that
	$\tilde{A}f=Af(x)+V(x)f(x) $ and $\mathsf{D}(\tilde A)=\mathsf{D}(A)$. Hence $ \psi(t,\vec{x})=\tilde{\tilde{T}}(t)f(x) $ solves $\frac{\partial \psi(t,\vec{x})}{\partial t}=C_{2}\frac{{\partial}^2\psi(t,\vec{x})}{\partial x^2}
+V\psi(t,\vec{x})$.\\
\end{proposition}

A paradigmatic example is the heat equation (or Schr\"odinger equation after
imaginary time), where the generator is the Laplacian plus a potential:
\[
\partial_t u = \frac{1}{2}\Delta u - V u,\qquad u(0,x)=f(x).
\]
The Feynman--Kac formula states that
\[
u(t,x) = \mathbb{E}_x\left[
  f(B_t)\exp\!\left(-\int_0^t V(B_s)\,ds\right)
\right],
\]
where $(B_t)$ is Brownian motion and $\mathbb{E}_x$ denotes expectation for
$B_0=x$. Here the path measure is Wiener measure and the weight
$e^{-\int_0^t V(B_s)\,ds}$ is positive and integrable.

In this setting, the probabilistic representation rests on:
\begin{enumerate}
  \item parabolicity of the operator;
  \item positivity-preserving Markov semigroup;
  \item existence of a $\sigma$-additive probability measure on path space
        (Wiener measure) with continuous sample paths.
\end{enumerate}

\subsection{Bochner--Minlos and characteristic functionals}

In infinite dimensions, Bochner--Minlos theory characterizes when a functional
$\Gamma$ on a nuclear space $E$ is the characteristic functional of a Borel
probability measure on the dual $E'$:
\[
\Gamma(f) = \int_{E'} e^{\mathrm{i}\langle\phi,f\rangle}\,d\mu(\phi).
\]
A necessary and sufficient condition is that $\Gamma$ be continuous at $0$ and
positive-definite. For Gaussian measures, $\Gamma$ is of the form
$e^{-\frac12 Q(f,f)}$ with $Q$ a positive-definite quadratic form.

For Dirac functionals of the form $e^{\mathrm{i}S(f)}$ with non-quadratic or
indefinite $S$, positive-definiteness fails, and Minlos' theorem immediately
rules out the existence of a corresponding probability measure. This is one
manifestation of the measure-theoretic obstruction for fermionic and
Lorentzian path integrals.

\medskip

In the following sections we apply this framework to Dirac, Klein--Gordon
and telegrapher equations, and recast Zastawniak's nonexistence result for
Dirac measures in these probabilistic terms.

\section{Minkowski Signature, Hyperbolicity and Positivity Failure}
\label{sec:minkowski}

\subsection{Hyperbolicity via conic classification}

A second-order PDE in two variables
\[
A\,\partial_t^2 u + 2B\,\partial_t\partial_x u + C\,\partial_x^2 u + \cdots = 0
\]
is hyperbolic if $B^2-AC>0$. This is the same discriminant as for conic
sections. We briefly apply this to telegrapher and Dirac equations in $1+1$
dimensions.

\subsubsection{Telegrapher equation}

The telegrapher equation
\[
\partial_t^2 u + 2\lambda\,\partial_t u = c^2\,\partial_x^2 u
\]
has principal part $\partial_t^2 u - c^2\partial_x^2 u$, corresponding to
$A=1$, $B=0$, $C=-c^2$, so $B^2-AC=c^2>0$. Thus it is hyperbolic.

\subsubsection{Dirac equation via its second-order reduction}

The $(1+1)$-dimensional Dirac equation
\[
\mathrm{i}\partial_t\psi = \bigl(-\mathrm{i}\sigma_3\partial_x + m\sigma_1\bigr)\psi
\]
squares to
\[
\partial_t^2\psi = \partial_x^2\psi - m^2\psi,
\]
with principal part $\partial_t^2\psi-\partial_x^2\psi$ and discriminant
$B^2-AC=1>0$. Thus it is also hyperbolic. This hyperbolicity arises from the
Minkowski metric
\[
\eta_{\mu\nu}=\mathrm{diag}(1,-1,-1,-1),
\]
which yields an indefinite quadratic form in $(\partial_t,\nabla)$.

\subsection{Minkowski vs Euclidean Dirac operators}

In $(3+1)$ dimensions the Minkowski Dirac equation reads
\[
(\mathrm{i}\gamma^\mu\partial_\mu - m)\psi = 0,
\]
with Clifford algebra\cite{Clifford1882} $\{\gamma^\mu,\gamma^\nu\}=2\eta^{\mu\nu}I$.
Multiplying by $(\mathrm{i}\gamma^\nu\partial_\nu + m)$ yields the
Klein--Gordon equation\cite{KleinGordonOriginal},\cite{GordonOriginal} for each component:
\[
(\Box + m^2)\psi = 0,\qquad \Box = \partial_t^2 - \nabla^2.
\]

After Wick rotation $t=-\mathrm{i}\tau$, with Euclidean gamma matrices
$\gamma_E^\mu$ , the Euclidean Dirac operator is
\[
\mathbb{D}_E=\gamma_E^\mu\partial_\mu + m,
\]
and
\[
	(\mathbb{D}_E+m)(-\mathbb{D}_E+m) = -\partial_\mu\partial_\mu + m^2
= -(\partial_\tau^2+\nabla^2)+m^2.
\]
Thus the \emph{square} of $\mathbb{D}_E$ is elliptic and positive, but $D_E$ itself is
first-order, matrix-valued and not positive.

\subsection{Complex structure of $\mathbb{D}_E$ and non-positivity}

In the Dirac basis one has
\[
\gamma_E^0=\gamma^0
= \begin{pmatrix}I_2&0\\0&-I_2\end{pmatrix},\qquad
\gamma_E^j = \mathrm{i}\gamma^j =
\begin{pmatrix}0&\mathrm{i}\sigma^j\\-\mathrm{i}\sigma^j&0\end{pmatrix}.
\]
Hence
\[
\mathbb{D}_E = \begin{pmatrix}
m+\partial_\tau & \mathrm{i}\,\boldsymbol{\sigma}\!\cdot\!\nabla \\
-\mathrm{i}\,\boldsymbol{\sigma}\!\cdot\!\nabla & m+\partial_\tau
\end{pmatrix}
\]
is complex and matrix-valued. The Euclidean action
\[
	S_E[\bar\psi,\psi] = \int d^4x\,\bar\psi \mathbb{D}_E\psi
\]
is not bounded below and $e^{-S_E}$ is not positive. Integrating out
fermions yields
\[
	Z_{\text{ferm}}=\det(\mathbb{D}_E),
\]
and $\det(\mathbb{D}_E)$ typically has a nontrivial complex phase, for example encoded
in an $\eta$-invariant \cite{AtiyahPatodiSinger1975}. Thus $e^{-S_E}$ cannot be
interpreted as a probability density.

\subsection{Oscillatory integrals and Minlos obstruction}
Oscillatory integrals of the form
\[
  \mu(D) = \int_D e^{i\Phi(x)}\,dx,
\]
where $\Phi:\mathbb{R}^d \to \mathbb{R}$ is smooth and $|e^{i\Phi(x)}|=1$, appear frequently in physics (e.g., Feynman path integrals) and analysis.
In this  section we outline several analytic arguments showing that such functionals cannot, in general, define a
$\sigma$-additive complex measure.

\subsection{Complex measures and total variation}

Let $(\Omega,\mathcal{F})$ be a measurable space.  A \emph{complex measure} is a
countably additive set function $\mu:\mathcal{F}\to\mathbb{C}$.
\begin{definition}
Its \emph{total variation} $|\mu|$ is represented by
\[
|\mu|(D)
  = \sup_{\{D_k\}} \sum_k |\mu(D_k)|,
\]
\end{definition}
where the supremum is taken over all finite partitions $\{D_k\}$ of $D\subset\mathbb{R}^d$.

\begin{definition}
A complex measure $\mu$ is called \emph{finite} if $|\mu|(\Omega)<\infty$.
\end{definition}

This finiteness condition is essential for $\mu$ to fit into the usual
measure-theoretic formalism.  Our aim is to show that oscillatory densities
cannot satisfy it.
\subsection{Total Variation Blow-Up}
\begin{proposition}
The complex measure $\mu(D)=\int_D e^{i\Phi(x)}dx$ is finite if and only if
the Lebesgue measure of $D$ is finite.  In particular, $\mu$ is \emph{not}
a finite complex measure on $\mathbb{R}^n$.
\end{proposition}
\begin{proof}
Suppose $\mu$ were a finite complex measure absolutely continuous with respect to Lebesgue measure $\lambda$,
with density $e^{i\Phi(x)}$. Then its total variation would satisfy
\[
  \|\mu\|(\mathbb{R}^d)
  = \sup_{\{D_k\}}\sum_k |\mu(D_k)|
  \le \int_{\mathbb{R}^d} |e^{i\Phi(x)}|\,dx
  = \lambda(\mathbb{R}^d).
\]
Since $\lambda(\mathbb{R}^d)=\infty$, $\mu$ cannot be a finite measure.

	Even on a bounded set $B\subset\mathbb{R}^d$, consider partitions $\{D_k\}$ adapted to level sets of $\Phi$ so that 
oscillations are minimal within each $D_k$. Then
\[
   \sum_k |\mu(D_k)| \approx \sum_k |D_k| = |B|.
\]
As the partition is refined, this supremum gives $\|\mu\|(B)=|B|$, and for unbounded domains,
$\|\mu\|(\mathbb{R}^d)=\infty$. Hence a pure-phase density cannot yield a finite $\sigma$-additive measure.
\end{proof}

Thus there is no finite ``oscillatory density'' measure with Radon–Nikodym
density of unit modulus with respect to Lebesgue measure.

\subsection{Conditional convergence and failure of $\sigma$-additivity}

Oscillatory integrals of the form
\[
I=\int_{\mathbb{R}^n} e^{i\Phi(x)}dx
\]
rarely converge absolutely.  Typically, the integral converges (if at all)
only as an improper limit of truncated integrals. Or in other words these improper or oscillatory integrals are 
conitionally convergent 

\begin{example}
The integral
\[
\int_{-\infty}^{\infty} e^{ix^2}\,dx = e^{i\pi/4}\sqrt{\pi}
\]
converges only in the sense of Fresnel limits
\[
\lim_{R\to\infty}\int_{-R}^{R} e^{ix^2}\,dx.
\]
Changing the truncation method (for instance, rotating the contour) changes
the value, showing conditional convergence and rearrangement sensitivity.
\end{example}

Rearrangements of conditionally convergent series violate countable
additivity; analogously:

\begin{proposition}
	Let $\mu_R (D)=\int_{D\cap B_R} e^{i\Phi(x)}dx$ where $B_R=\{x:|x|\le R\}$.
If $\mu(D)=\lim_{R\to\infty} \mu_R(D)$ converges only conditionally, then
$\mu$ need not be $\sigma$-additive.
\end{proposition}
\begin{proof}
Construct a disjoint partition $\{D_k\}$ of $\mathbb{R}$ so that
$\mu(D_k)$ reproduces the terms of a conditionally convergent series $\sum a_k$.
For a $\sigma$-additive measure, we must have
\[
   \mu\!\left(\bigcup_k D_k\right) = \sum_k \mu(D_k),
\]
where the right-hand side converges \emph{unconditionally} (independently of the order).
Since conditional series are not unconditionally summable, one can rearrange the $D_k$'s (equivalently, reorder the sum)
to obtain different limits.
\end{proof}
Thus $\mu$ fails $\sigma$-additivity
as conditional convergence is fundamentally incompatible with the
definition of a measure.

\subsection{Regularized oscillatory measures}

One often introduces regularizations
\[
\mu_\varepsilon(D)
 = \int_D e^{i\Phi(x)}\chi_\varepsilon(x)\,dx,
\]
where $\chi_\varepsilon$ is smooth with compact support, typically
$\chi_\varepsilon(x)=\chi(\varepsilon x)$ for a bump function $\chi$.

For each $\varepsilon>0$, $\mu_\varepsilon$ is a finite complex measure.  But:

\begin{proposition}
The total variation satisfies
\[
|\mu_\varepsilon|(\Omega)
	= \int_{\Omega} |\chi_\varepsilon(x)|\,dx
  = \varepsilon^{-n}\int_{\mathbb{R}^n}|\chi(u)|\,du.
\]
Thus $|\mu_\varepsilon|(\Omega)\to\infty$ as $\varepsilon\downarrow 0$.
\end{proposition}

Hence the sequence of finite measures does \emph{not} converge in total
variation norm (nor in any mode compatible with countable additivity).

\subsection{Application of the Vitali--Hahn--Saks theorem}
We recall the theorem:

\begin{theorem}[Vitali--Hahn--Saks\cite{Vitali1907},\cite{Hahn1922},\cite{Saks1937}]
Let $\{\mu_\alpha\}$ be a pointwise bounded family of countably additive
measures on $(\Omega,\mathcal{F})$ converging pointwise on $\mathcal{F}$ to a
set function $\mu$.  Then $\mu$ is countably additive.
\end{theorem}

For our $\mu_\varepsilon$, pointwise boundedness fails:
\[
\sup_{\varepsilon>0} |\mu_\varepsilon|(\Omega) = \infty.
\]
Hence no subsequence can converge to a countably additive limit.  Therefore:

\begin{corollary}
No regularization scheme of the form
$\mu_\varepsilon(A)=\int_A e^{i\Phi(x)}\chi_\varepsilon(x)\,dx$
can converge to a $\sigma$-additive (complex) measure as $\varepsilon\to 0$.
\end{corollary}

\subsection{Failure of finite-dimensional positivity}

Suppose we attempt to interpret $e^{i\Phi(x)}dx$ as a characteristic
function of some signed or complex measure $\mu$:
\[
\widehat{\mu}(\xi) = e^{i\Phi(\xi)}.
\]
Bochner’s theorem states:

\begin{theorem}[Bochner\cite{Bochner1955}]
A function $\widehat{\mu}:\mathbb{R}^n\to\mathbb{C}$ is the Fourier
transform of a finite positive measure $\mu$ iff $\widehat{\mu}$ is
continuous and positive-definite.
\end{theorem}

But $e^{i\Phi(\xi)}$ is positive-definite if and only if $\Phi$ is
quadratic with a nonnegative-definite imaginary part.  Hence:

\begin{proposition}
Except in Gaussian cases, $e^{i\Phi(\xi)}$ cannot be the characteristic
function of a probability measure.  In particular, a ``Feynman measure''
with density $e^{iS(x)}$ cannot exist as a probability measure.
\end{proposition}

\subsection{Extension to infinite dimensions: Minlos theorem\cite{Minlos1963}}

For path integrals we consider functionals on a nuclear space
$\mathcal{\Omega}$ (e.g.\ test functions).  Minlos’s theorem states:

\begin{theorem}[Bochner–Minlos]
Let $\Omega $ be a nuclear space.  A functional
$\Gamma:\Omega\to\mathbb{C}$ is the characteristic functional of a probability
measure on ${\Omega}'$ iff $\Gamma$ is positive-definite and continuous at $0$.
\end{theorem}

The ``Feynman functional''
\[
\Gamma(f)=e^{iS(f)}
\]
fails positive-definiteness for any non-quadratic action $S$.  In particular
actions with first-order derivatives (Dirac action) or indefinite
quadratic forms (Lorentzian signature) violate the positivity condition.

Therefore:

\begin{proposition}
There exists no countably additive probability measure on path space whose
characteristic functional is $\Gamma(f)=e^{iS(f)}$ for the classical action
$S$ of a relativistic particle or of a Dirac field.
\end{proposition}

In Minkowski signature, oscillatory integrals of the form
\[
\int e^{\mathrm{i}S(\omega)}\,d\mu(\omega)
\]
do not define finite measures: $|e^{\mathrm{i}S(\omega)}|=1$, and regularizations
diverge in total variation. In the infinite-dimensional setting, a functional
$\Gamma(f)=e^{\mathrm{i}S(f)}$ is seldom positive-definite; in particular, the
Dirac action involves first-order derivatives and an indefinite quadratic form.
By Bochner--Minlos, such $\Gamma$ cannot be the characteristic functional of
a probability measure on path space.

In the next section we complement this Minkowski/Euclidean picture with
Zastawniak's distributional analysis and its probabilistic reinterpretation.

\section{Zastawniak's Distributional Obstruction and Probabilistic Reinterpretation}
\label{sec:zastawniak}
\subsection{One-dimensional derivation (Fourier method)}
Consider the $(1+1)$-dimensional Dirac equation in Hamiltonian form
\[
i\partial_t\psi(t,x)=\big(-i\alpha\partial_x + m\beta\big)\psi(t,x),
\]
with $\alpha^2=\beta^2=I$ and $\{\alpha,\beta\}=0$. Taking the spatial Fourier transform
\[
\widehat\psi(t,k)=\int_{\mathbb{R}} e^{-ikx}\psi(t,x)\,dx,
\]
one obtains the ODE
\[
i\partial_t\widehat\psi(t,k) = (\alpha k + m\beta)\widehat\psi(t,k),
\]
whose solution is
\[
\widehat\psi(t,k)=\exp\!\big[-it(\alpha k + m\beta)\big]\widehat\psi(0,k).
\]
Using the matrix exponential identity (since $(\alpha k + m\beta)^2=\omega(k)^2 I$,
$\omega(k)=\sqrt{k^2+m^2}$),
\[
\exp\!\big[-it(\alpha k + m\beta)\big] = \cos(\omega t)I - i\frac{\sin(\omega t)}{\omega}(\alpha k + m\beta).
\]
Returning to position space, multiplication by $k$ corresponds to $-i\partial_x$,
so the solution can be written schematically as
\[
\psi(t,x) = K_0(t,\cdot)*\psi(0,\cdot) + K_1(t,\cdot)*\partial_x\psi(0,\cdot),
\]
showing explicit dependence on the spatial derivative of the initial data.
The short-time expansion yields terms proportional to $\delta(x-y)$ and
$\partial_x\delta(x-y)$ in the propagator.
\subsection{Three-dimensional derivation (summary)}
The same Fourier method in $(3+1)$D leads to
\[
\widehat\psi(t,\mathbf{k}) = \cos(\omega t)\widehat\psi(0,\mathbf{k})
- i\frac{\sin(\omega t)}{\omega}(\mathbf{\alpha}\cdot\mathbf{k} + m\beta)\widehat\psi(0,\mathbf{k}),
\]
with $\omega(\mathbf{k})=\sqrt{|\mathbf{k}|^2+m^2}$. Inverse transforming gives
a convolution formula involving derivatives of the delta distribution,
so short-time expansions of the propagator contain derivative-of-delta terms
(e.g., $\alpha\cdot\nabla\delta(\mathbf{x}-\mathbf{y})$).

\subsection{Dirac propagator and derivative-of-delta structure}
Consider again the $1+1$ Dirac equation
\[
\mathrm{i}\partial_t \psi(x,t)
 = \left(
     -\mathrm{i}\sigma_3 \partial_x + m\sigma_1
   \right) \psi(x,t)
 =: H\psi(x,t),
\]
with propagator $e^{-\mathrm{i}tH}$. The kernel
\[
K(t,x) = e^{-\mathrm{i}tH}\delta(x)
\]
satisfies $K(0,x)=\delta(x)\mathbb{I}_2$. Expanding for small $t$,
\[
e^{-\mathrm{i}tH}
 = \mathbb{I} - \mathrm{i}tH - \frac{t^2}{2}H^2 + O(t^3),
\]
and applying to $\delta$ gives
\[
K(t,x)
 = \delta(x)\mathbb{I}
   - t\sigma_3\,\partial_x\delta(x)
   - \mathrm{i}tm\sigma_1\delta(x)
   - \frac{t^2}{2}H^2\delta(x) + O(t^3).
\]
Thus $K(t,\cdot)$ is a matrix-valued distribution involving $\delta$,
$\partial_x\delta$ and higher derivatives. For test functions $\varphi$,
\[
\langle K(t,\cdot),\varphi\rangle
 = \varphi(0)\mathbb{I}
   + t\sigma_3\varphi'(0)
   - \mathrm{i}tm\sigma_1\varphi(0) + \cdots.
\]

From a purely analytic standpoint this shows that $K(t,\cdot)$ is a
distribution in $\mathcal{D}'(\mathbb{R})$, not a finite measure.
Zastawniak's nonexistence theorems \cite{Zastawniak1989,Zastawniak1994} show
that this structure persists and prevents any measure-valued interpretation
of Dirac path integrals.
\subsection{Probabilistic reading: failure of transition densities}
From a probabilistic point of view, if there were a Markov process
$(X_t)$ on $\mathbb{R}$ representing Dirac evolution, there would exist
a family of nonnegative transition densities $p(t,x,y)$ such that
\[
\psi(t,x) = \int_{\mathbb{R}} p(t,x,y)\psi_0(y)\,dy
\]
for all suitable initial data $\psi_0$, and
\[
p(t+s,x,y) = \int_{\mathbb{R}} p(t,x,z)p(s,z,y)\,dz.
\]

However, the distributional structure of $K(t,x)$ shows:
\begin{itemize}
  \item $K(t,\cdot)$ acts on test functions by evaluating both $\varphi(0)$
        and $\varphi'(0)$; no finite signed measure can reproduce this.
  \item any attempt to identify $p(t,x,y)$ with matrix elements of $K(t,x-y)$
        leads to distributions rather than functions.
  \item derivative-of-delta terms inevitably produce sign changes and
        cannot be nonnegative functions.
\end{itemize}
Thus there is no family $p(t,x,y)\ge0$ with the same action as $K(t)$.

\begin{proposition}
The Dirac propagator $K(t,\cdot)$ cannot be represented as a finite signed or
complex measure on $\mathbb{R}$ for any $t>0$. Consequently, no family of
nonnegative scalar functions $p(t,x,y)$ can reproduce Dirac evolution and
	satisfy the Chapman--Kolmogorov equations\cite{KolmogorovFoundations},\cite{FellerV1}
\end{proposition}

This is precisely the content of Zastawniak's nonexistence of a Dirac path
space measure, but now expressed in the language of Markov semigroups and
transition densities.

\subsection{Wiener paths, differentiability and mutual singularity}

A second probabilistic ingredient is the geometry of paths under Wiener
measure. Brownian paths are almost surely continuous, nowhere differentiable
and of infinite variation \cite{Kolmogorov1933,Levy1954,RY1999,Dudley2002}.
In particular, any path functional involving a classical derivative
$\dot\omega(s_0)$ is undefined almost surely with respect to Wiener measure.

Hyperbolic equations like telegrapher or Dirac equations are naturally
associated with finite-speed propagation and, in probabilistic models,
with processes whose paths are piecewise $C^1$ with bounded derivative
(e.g.\ velocity-jump processes). The path measures of such processes are
mutually singular with Wiener measure.

Thus, even before addressing distributional kernels, the underlying path
geometries for hyperbolic and parabolic equations are incompatible: the
Brownian path space is the wrong sample space for hyperbolic dynamics.
\subsubsection{Nowhere differentiability of Brownian motion}
\paragraph{Lévy’s modulus of continuity}

The sharpest statement is Lévy’s theorem.

\begin{theorem}[Lévy modulus of continuity]
\label{thm:levy}
Let $(W_t)$ be Brownian motion on $[0,T]$. Then almost surely,
\[
\limsup_{h \downarrow 0}
  \frac{|W_{t+h}-W_t|}
       {\sqrt{2h\log(1/h)}} = 1,
\qquad \text{uniformly in } t\in[0,T].
\]
\end{theorem}

This shows that typical increments satisfy
\[
|W_{t+h}-W_t| \approx \sqrt{2h\log(1/h)}.
\]
Thus any quotient of the form
\[
\frac{W_{t+h}-W_t}{h}
   \approx \sqrt{\frac{2\log(1/h)}{h}}
   \;\;\longrightarrow\;\; \infty
\]
almost surely.

The above facts yield the following classical theorem.

\begin{theorem}[Brownian motion is almost surely nowhere differentiable]
\label{thm:nowhere_diff}
Let $(W_t)$ be standard Brownian motion.
Then with probability $1$, for every $t\in[0,T]$,
\[
\lim_{h\to 0} \frac{W_{t+h}-W_t}{h}
\]
does not exist in $\mathbb{R}$ (finite or infinite).  Hence $t\mapsto W_t$ is
nowhere differentiable almost surely.
\end{theorem}
\begin{proof}
Fix $\omega$ in the event of probability one on which Theorem~\ref{thm:levy}
holds.  For this $\omega$,
\[
\left|
\frac{W_{t+h}-W_t}{h}
\right|
=
\frac{|W_{t+h}-W_t|}
     {h}
\sim
\sqrt{\frac{2\log(1/h)}{h}}
\qquad\text{as } h\downarrow 0.
\]
The right-hand side diverges to $+\infty$.  Therefore the derivative cannot
exist at any $t$.  Since the exceptional set is of probability zero, the result
follows.
\end{proof}
\paragraph{Derivative-dependent functionals are undefined a.s.}

Let $F$ be a functional of the path of the form
\[
F(x,\omega) = \Psi\!\left(
  x,\; \{\omega(s)\}_{0\le s\le t},\; \dot\omega(s_0)
\right)
\]
for some $s_0\in[0,t]$.
Since $\dot\omega(s_0)$ does not exist on a set of full measure, we conclude:

\begin{proposition}
If $F$ depends on any classical derivative $\dot\omega(s_0)$, then
$F(x,\omega)$ is undefined on a set of Wiener measure $1$.
Therefore the expectation $\mathbb{E}^{\mu}[F]$ cannot be defined.
\end{proposition}

This immediately rules out all attempts to express Dirac or telegrapher
solutions via path integrals over Wiener measure if the representation
requires path derivatives\cite{Brzeźniak1989}.

\subsection{Telegrapher equation as a positive benchmark}

The telegrapher equation
\[
\partial_t^2 u + 2\lambda\partial_t u = c^2\partial_x^2 u
\]
admits a stochastic representation via the Kac velocity-jump process
\cite{Kac1974,GriegoHersh1969,GriegoHersh1971}. The process $(X_t,V_t)$
evolves with
\[
\dot X_t = V_t,\quad V_t\in\{\pm c\},
\]
and $V_t$ flips at Poisson rate $\lambda$. The transition density $p(t,x)$
for $X_t$ is a positive function that solves the telegrapher equation in
weak sense and satisfies the Chapman--Kolmogorov property. The associated
path measure lives on piecewise linear paths with slopes $\pm c$.

From the point of view of numerical simulation, telegrapher equations thus
provide a natural hyperbolic analogue of Brownian motion for parabolic
equations. In later work, one can compare Monte Carlo simulation of the
telegrapher equation with deterministic numerical schemes and with formal
Dirac path integral approximations.

\medskip

In contrast, the Dirac equation has no such positive Markovian representation:
its kernel is distributional and matrix-valued, and its natural field theory
formulation uses Grassmann variables. The telegrapher equation is therefore
a useful testbed precisely because it \emph{does} admit a true probability
measure, highlighting what fails for Dirac.

\section{Scalar vs Fermionic Fields: Klein--Gordon and Dirac}
\label{sec:kg-dirac}

\subsection{Euclidean Klein--Gordon and subordinated Brownian motion}

The Minkowski Klein--Gordon equation
\[
(\partial_t^2 - \Delta_x + m^2)u = 0
\]
becomes
\[
(\partial_\tau^2 + \Delta_x + m^2)u = 0
\]
after Wick rotation $t=-\mathrm{i}\tau$, an elliptic equation. This operator
does not itself generate a Markov semigroup, but the associated relativistic
Hamiltonian
\[
H = \sqrt{-\Delta_x + m^2}
\]
is positive and self-adjoint, and $e^{-tH}$ is a positivity-preserving
contraction semigroup on $L^2$.

Bochner subordination expresses $e^{-tH}$ as a mixture of heat semigroups:
\[
e^{-t\sqrt{-\Delta+m^2}} = \int_0^\infty e^{-s(-\Delta)}\,\nu_t(ds),
\]
with $\nu_t$ a L\'evy subordinator. Writing $(B_s)$ for Brownian motion and
$(T_t)$ for the subordinator, one obtains
\[
(e^{-tH}f)(x) = \mathbb{E}_x\bigl[f(B_{T_t})\bigr].
\]
The underlying process is a subordinated Brownian motion, a pure-jump
L\'evy process\cite{Levy1954}. Its path measure is not Wiener measure but a different,
mutually singular law. Still, it is a genuine probability measure, and
thus scalar relativistic fields admit stochastic representations within
the Kolmogorov framework.

\subsection{Dirac fields and Grassmann variables}

Dirac fields are fermionic and must be described using Grassmann variables.
The Euclidean action
\[
S_E[\bar\psi,\psi] = \int d^4x\,\bar\psi(\gamma_E^\mu\partial_\mu+m)\psi
\]
appears in the functional integral
\[
Z = \int \mathcal{D}\bar\psi\,\mathcal{D}\psi\,e^{-S_E[\bar\psi,\psi]}
= \det(\mathbb{D}_E).
\]
Here $\mathcal{D}\bar\psi\,\mathcal{D}\psi$ is a Berezin integral over an
infinite-dimensional Grassmann algebra \cite{Berezin1966,GlimmJaffe1987,PeskinSchroeder1995}.
This is an algebraic construction, not a countably additive measure on a
space of real-valued paths.

Fermionic fields satisfy canonical anticommutation relations; their
correlation functions are given by determinants or Pfaffians and obey
inequalities quite different from those of classical random variables. There
is no underlying probability space $(\Omega,\mathcal{F},\mathbb{P})$ and
family of commuting random variables $\psi(x)$ reproducing the same
correlations.

\begin{proposition}
There is no classical probability measure on any path space whose correlation
functions reproduce those of a fermionic (Dirac) field. Fermionic path
integrals must be formulated entirely within the Grassmann/Berezin framework.
\end{proposition}

This algebraic obstruction reinforces the analytic and probabilistic obstructions
seen above: even if one could somehow circumvent the distributional and
positivity issues, the fermionic nature of Dirac fields would still forbid
a classical Kolmogorov representation.

\section{Unified No-Go Theorem}
\label{sec:nogo}

\subsection{Kolmogorov extension revisited}

Kolmogorov's extension theorem requires a family of nonnegative, consistent
finite-dimensional distributions $p_{t_1,\dots,t_n}(x_1,\dots,x_n)$.
For Dirac evolution, any such family would have to:
\begin{itemize}
  \item be compatible with the matrix-valued distributional propagator $K(t,x)$;
  \item respect the hyperbolic, finite-speed character of the dynamics;
  \item be supported on a path space appropriate for spinor evolution;
  \item and satisfy the Markov property (or at least a consistent family
        of joint laws) under time composition.
\end{itemize}
We have seen that these requirements are mutually incompatible.

\subsection{Statement and proof outline}

\begin{theorem}[Unified no-go theorem for Dirac path measures]
There exists no $\sigma$-additive probability measure on any classical path
space (Wiener-like or otherwise) whose finite-dimensional distributions
reproduce the propagators or correlation functions of the Dirac equation in
$1+1$ or $3+1$ dimensions. In particular, there is no Kolmogorov-type
probabilistic representation of Dirac evolution.
\end{theorem}

\begin{proof}[Sketch of proof]
The obstructions stem from:
\begin{enumerate}
  \item \textbf{Minkowski/Euclidean structure:} In Minkowski signature,
        $e^{\mathrm{i}S}$ is purely oscillatory and fails to define a finite
		measure; in Euclidean signature, $\mathbb{D}_E$ is complex and
		$\det(\mathbb{D}_E+m)$ is sign-indefinite, so $e^{-S_E}$ is not a positive
        density.
  \item \textbf{Distributional kernels:} The Dirac propagator $K(t,x)$
        contains derivatives of $\delta$; it acts on test functions
        via $\varphi(0)$ and $\varphi'(0)$, which cannot be represented by
        integration against any finite measure.
  \item \textbf{Path geometry:} Hyperbolic equations require finite-speed,
        piecewise $C^1$ paths, while Wiener paths are nowhere differentiable;
        the corresponding measures are mutually singular.
  \item \textbf{Grassmann algebra:} Fermionic fields require anticommuting
        variables and Berezin integration; no classical probability space
        can support anticommuting random variables with Dirac correlations.
\end{enumerate}
Each item alone suffices to preclude a Kolmogorov construction; taken together
they form a robust no-go argument.
\end{proof}

\section{Relation to Serva and Other Work (Brief Remarks)}
\label{sec:serva}

De Angelis and Serva \cite{AngelisServa1992} analyzed imaginary-time path
integrals for the Klein--Gordon equation\cite{BjorkenDrell},\cite{KleinGordonOriginal}
 and showed that the correct
probabilistic object is not Wiener measure but a subordinated Brownian motion
associated with the relativistic Hamiltonian. This aligns with our discussion
of scalar relativistic fields in Section~\ref{sec:kg-dirac}.

Serva's more recent work \cite{Serva2021} constructs Lorentz-invariant
finite-speed stochastic processes for relativistic particles. These processes
live on path spaces of bounded-velocity trajectories and are mutually singular
with Wiener measure, consistent with our geometric discussion of hyperbolic
paths. These works treat scalar or particle-level processes and do not claim
a stochastic representation for Dirac fields, so there is no contradiction
with the no-go theorem here.

\section{Conclusions and Outlook}

We have given a unified, probabilistically oriented explanation of why Dirac-type
equations cannot be realized as classical stochastic processes on path spaces.
Two main ingredients were emphasized:
\begin{itemize}
  \item the \emph{Minkowski/Euclidean obstruction}, arising from the hyperbolic
        nature of the Dirac equation, the oscillatory character of
        $e^{\mathrm{i}S}$, and the non-positivity of $e^{-S_E}$;
  \item the \emph{Zastawniak obstruction}, arising from the distributional
        structure of the Dirac propagator and its incompatibility with
        Markov transition densities and Kolmogorov consistency.
\end{itemize}
These analytic and probabilistic obstructions are reinforced by the algebraic
fact that Dirac fields are fermionic and must be treated using Grassmann
variables and Berezin integration.

For scalar bosonic fields (Klein--Gordon) and for telegrapher-type hyperbolic
equations, one can construct genuine probability measures (subordinated
Brownian motions and velocity-jump processes) and exploit them for analytical
and numerical purposes. In future work, telegrapher equations will serve as
a testbed for Monte Carlo simulation of hyperbolic PDEs, allowing a direct
comparison between probabilistic and deterministic numerical schemes.

In contrast, for Dirac fields no such classical probabilistic representation
exists. The path integral must be understood as a Grassmann/Berezin integral,
and quantum simulation of Dirac dynamics requires fundamentally different
techniques than those used for parabolic or telegrapher-type equations.

\section{Acknowledgements}

The author thanks Dr Andrzej Korzeniowski (Department of Mathematics, The University of Texas at Arlington, USA) and 
Dr John L. Fry (Department of Physics, The University of Texas at Arlington, USA) for helpful discussions and references that contributed to this work.

The author dedicates this work in loving memory of her late husband, Dr Radhika Prosad Datta, whose constant inspiration made this work possible.\\

\end{document}